\documentclass[a4paper,UKenglish]{lipics}
\usepackage{enumerate,amsfonts,amssymb,amsmath,tikz}

\definecolor{darkgreen}{rgb}{0,0.4,0}
\newcommand{\todo}[1]{}

\theoremstyle{plain}

\makeatletter
\@ifundefined{mfcs}{%
\newtheorem{theorem}{Theorem}

\renewtheorem{prop}[theorem]{Proposition}

\theoremstyle{definition}

}{}
\makeatother

\newcommand{\str}[1]{{\mathfrak{#1}}}

\newcommand{\naw}[1]{{\langle #1\rangle}}

\newcommand{\PSPACE}{{\sc PSpace}}

\newcommand{\NP}{{\sc NP}}
\newcommand{\EXPTIME}{{\sc ExpTime}}

\newcommand{\klasa}{{\cal K}}

\makeatletter
\def\blfootnote{\xdef\@thefnmark{}\@footnotetext}
\makeatother

\author[1]{Jakub Michaliszyn}
\affil[1]{University Of Wroc\l{}aw}

\title{Elementary Multimodal Logics}

\authorrunning{Jakub Michaliszyn}

\Copyright{J. Michaliszyn}%mandatory. LIPIcs license is "CC-BY";  http://creativecommons.org/licenses/by/3.0/

% mandatory: Please provide 1-5 keywords
\title{Elementary Multimodal Logics}

\begin{document}

\maketitle
  \begin{abstract}
We study multimodal logics over universally first-order definable classes of frames.
We show that even for bimodal logics, there are universal Horn formulas that define set of frames such that the satisfiability problem is undecidable, even if one or two of the binary relations are transitive.

  \end{abstract}
\blfootnote{This work has been supported by Polish National Science Center grant UMO-2014/15/D/ST6/00719.}

\section{Introduction}
Multimodal logic extends propositional logic by unary operators: $\Diamond_1, \dots, \Diamond_m$ and dual operators $\square_1, \dots, \square_m$. The formal semantics is given in terms of Kripke structures. 
A Kripke structure is a structure consisting of a (possibly infinite) number of \emph{worlds} that may be connected by binary relations $R_1, \dots, R_m$, called a \emph{frame}, together with a valuation of propositional variables in each of the worlds.  
In this semantics, $\Diamond_i \varphi$ means  \emph{the current world is connected by $R_i$ to some world in which $\varphi$ is true}; and $\square_i \varphi$, equivalent to $\neg \Diamond_i \neg \varphi$, means \emph{$\varphi$ is true in all worlds to which the current world is connected by $R_i$}.  

Multimodal logics are often considered with some addidional constraints on the interpretation of operators, e.g. by requiring that the modal operator represents a relation that is reflexive and transitive (S4). 
In this way, we can define logics with nonuniform  modal operators, like Linear Temporal Logic (LTL), Computation Tree Logic (CTL) or Halpern--Shoham logic (HS).

Variants of multimodal logic vary in the complexity and the decidability of the satisfiability problem. Some multimodal logics are \NP{}-complete (S5), some are \PSPACE{}-complete (unimodal logic, LTL),  so are \EXPTIME{}-complete (CTL) and 2-\EXPTIME{}-complete (CTL${}^*$).
 Finally, the Halpern--Shoham logic is a simple example of a temporal logic that is undecidable, even if we consider some unimodal fragments \cite{BD11,MM11}.

In this paper we consider adding constraints on the interpretation of operators by defining the class of frames by a universal first order logic sentence that uses binary relations $R_1, \dots, R_n$. For example, the sentence $\forall x y z.(xR_1y \wedge yR_1z \Rightarrow xR_1z) \wedge (x R_2 y \Rightarrow x R_1 y)$ defines the class of all the frames where $R_1$ is a transitive relation containing $R_2$. Multiodal logic over a class of frames definable by a first order logic sentence is called \emph{an elementary multimodal logic}.

Interesting positive results regarding the decidability of various elementary \emph{unimodal} logic (i.e., a logic whose frames have only one binary relation) \cite{MichaliszynOtop15, csl13} give us hope that elementary multimodal logics can be used in verification, and, in particular, in synthesis. 

In this paper we focus on the synthesis problem, in which, for a given specification, the goal is to construct a system that satisfies the specification. The specification can be provided by any multimodal logic. However, most of the popular multimodal logic cannot express some possibly desired properties of the system, e.g., state that the system has to be deterministic (only one successor from each state), transitive (i.e., can do several steps at one), reflexive (i.e., can stay in a single state for a long time, perhaps because of not getting CPU time to proceed) etc. It makes therefore sense to consider a variant of the synthesis problem where properties like these can be easily stated. Elementary multimodal logics fit perfectly into this scenario.

Many modal logics used in automatic verification contain operators that are interpreted as transitive relations. For example, Linear Temporal Logic (LTL) contains transitive operators $F$ and $G$ \cite{modal}. In epistemic modal logics, the knowledge operators $K_i$ are interpreted as relations that are not only transitive, but also reflexive and symmetric \cite{F95}. Another example is the logic of subintervals \cite{MM11}, which is a fragment of Halpern--Shoham logic with a single modality that can be read as ``in some subinterval''.
\smallskip

{\bf Main results.}
In \cite{HS11}, it was shown that there is an universal first-order formula such that the global satisfiability problem of unimodal logic over the class of frames that satisfy this formula is undecidable. A slight modification of that formula yields an analogous result for the local satisfiability problem. In \cite{KMO11} it was shown that even a very simple formula with three variables without equality leads to undecidability. On the positive side, in \cite{csl13}, it was shown that universal formulas that imply that the (only) binary relation is transitive lead to decidable unimodal logic, and in \cite{MichaliszynOtop15} a similar result was shown for universal Horn formulas.

In this paper, we show that the positive results mentioned above do not extend to elementary multimodal logics. We show that elementary multimodal logics may be undecidable even if we consider only universal Horn formulas and two relations, regardless of whether we assume that any of the relations is transitive.

\section{Preliminaries}
Formulae of a multimodal logic are interpreted in  Kripke structures, which are triples of the form $\langle M, R_1, \dots, R_n, \pi \rangle$, where $M$ is a set of worlds, $R_1, 
\dots, R_n$ are binary relations on $M$,  $\langle M, R_1, \dots, R_n \rangle$ is called a \emph{frame}, and $\pi$ is \emph{a labelling}, a function that assigns to each world a set of propositional variables which are true at this world. We say that a structure  $\langle M, R, \pi \rangle$ is \emph{based} on the frame $\langle M, R \rangle$. For a given class of frames $\cal K$, we say that a structure is $\cal K$-based if it is based  on some frame from $\cal K$. We will use calligraphic letters ${\cal M}, {\cal N}$ to denote frames and Fraktur letters $\str{M}, \str{N}$ to denote structures. Whenever we consider a structure $\str M$, we assume that its frame is $\cal M$ and its universe is $M$ (and the same holds for other letters). 

The semantics of a multimodal logic is defined recursively. A modal formula $\varphi$ is (locally) \emph{satisfied} in a world $w$ of a model $\str M=\naw{M, R_1, 
dots, R_n, \pi}$,
denoted as
${\str  M}, w \models \varphi$ if
\begin{itemize}
\item $\varphi = p$, where  $p$ is a variable, and $\varphi \in \pi(w)$,
\item $\varphi = \neg \varphi'$ and ${\str  M}, w \not \models \varphi'$,
\item $\varphi = \varphi_1 \wedge \varphi_2$ and ${\str  M}, w \models \varphi_1$ and ${\str  M}, w \models \varphi_2$,
\item $\varphi = \Diamond_i \varphi'$ and there is a world $v\in M$ such that $(w, v) \in R_i$ and ${\str  M}, v \models \varphi'$,
\end{itemize}

Boolean connectives $\vee, \Rightarrow, \Leftrightarrow$ and constants $\top, \bot$ are introduced in the standard way. We abbreviate $\neg \Diamond_i \neg \varphi$ by $\square_i \varphi$. 
By $|\varphi|$ we denote the length of $\varphi$.

A formula $\varphi$ is \emph{globally} satisfied in ${\str  M} $, denoted as ${\str  M} \models \varphi$, if for all worlds $w$ of ${\str  M}$,  we have ${\str  M}, w \models \varphi$. 

For a given class of frames $\klasa$, we say that a formula $\varphi$ is \emph{locally} (resp.~\emph{globally}) ${\klasa}$-\emph{satisfiable} if there exists a $\klasa$-based structure ${\str M}$, and a world $w \in W$ such that ${\str M}, w \models \varphi$ (resp.~${\str M} \models \varphi$).  We study four versions of the satisfiability problem.

For a given formula $\varphi$, a Kripke structure $\str M$, and a world $w \in M$ we define the \emph{type} of $w$ (with respect to $\varphi$) in $\str M$ as $tp_{\str  M}^{\varphi}(w) = \{ \psi : {\str  M}, w \models \psi $ and $\psi$ is a subformula of $\varphi \}$. We write $tp_{\str  M}(w)$ if the formula is clear from the context. Note that $|tp_{{\str  M}}^{\varphi}(w)| \leq |\varphi|$. 

The class of 
 (equality-free) 
universal first order sentences is defined as a subclass of first--order sentences such that each sentence is of the form $\forall \vec{x} \,\Psi(\vec{x})$, where $\Psi(\vec{x})$ is quantifier--free formula over the language $\{ R_1, \dots, R_n \}$, where each $R_i$ is a binary relation symbol.

The \emph{local} (resp.~\emph{global}) \emph{satisfiability problem} $\klasa$-SAT (resp.~${\klasa}$-GSAT) as follows. Is a given modal formula locally (resp.~globally) ${\klasa}$-satisfiable? 
 The \emph{finite local} (global) satisfiability problem, $\klasa$-FINSAT ($\klasa$-GFINSAT), is defined in the same way, but we are only interested in finite models (the class $\klasa$ may still contain infinite structures). 

\section{Decidability of elementary multimodal logics}
Let $UFO(n; m)$ be the set of all the universal first order formulas $\Phi$ with binary relations $R_1, \dots, R_{n+m}$ such that in all the models of $\Phi$ the relations $R_{n+1}, \dots, R_{n+m}$ are transitive. 
Let $UHF(n; m)$ be a subset of $UFO(n; m)$ that consists of the Horn formulas.
We show the following.

\begin{theorem}
For any $n, m$ such that $n+m>1$, there is a $UHF(n; m)$ formula $\Phi$ such that the finite satisfiability, the finite global satisfiability, the satisfiability and  the global satisfiability of multimodal logic over $\Phi$ are undecidable.
\end{theorem}

\begin{proof}
We define $\Phi$ as the conjunction of the following formulas:
\begin{enumerate}
\item \(
\forall x, y, z, u, s, t . x R_1 y \wedge x R_1 u \wedge u R_1 z \wedge \mathit{mid_2}(u) \Rightarrow y R_2 z
\)
\item \(
\forall x, y, z, u, s, t . x R_2 y \wedge x R_2 u \wedge u R_1 z \wedge \mathit{mid_1}(u) \Rightarrow y R_1 z
\)
\item only if $n<2$: \(
\forall x, y, z . x R_2 y \wedge y R_2 z \Rightarrow x R_2 z
\) 
\item only if $n=0$: \(
\forall x, y, z . x R_1 y \wedge y R_1 z \Rightarrow x R_1 z
\) 
\end{enumerate}
where $\mathit{mid}_i(u)$ states that there is a path of length $2$ via $R_2$ starting from $u$: $u R_i s \wedge s R_i t$.

\usetikzlibrary{calc,matrix,arrows,math}
\tikzstyle{vertex}=[circle,draw=black,minimum size=30pt]
\tikzstyle{smallvertex}=[circle,draw=black,minimum size=15pt,inner sep=0pt]
\tikzstyle{dots}=[minimum size=15pt,inner sep=0pt]

\begin{figure}
\begin{tikzpicture}[scale=4]
\begin{scope}

%nodes P
\foreach \xx/\xxm in {0/0,1/1,2/2,3/0}
\foreach \yy/\yym in {0/0,1/1,2/2,3/0}
{
    \node[vertex] (P-\xx\yy) at (\xx, \yy) {$P_{\xxm, \yym}$};
}    

%nodes U, T, S
\foreach \xx/\xxm in {0/0,1/1,2/2}
\foreach \yy/\yym in {0/0,1/1,2/2}
{
    \node[smallvertex] (U-\xx\yy) at (\xx+0.5, \yy+0.5) {$U_{\xxm, \yym}$};
    \node[smallvertex] (S-\xx\yy) at (\xx+0.5, \yy+0.8) {$S_{\xxm, \yym}$};
    \node[smallvertex] (T-\xx\yy) at (\xx+0.2, \yy+0.8) {$T_{\xxm, \yym}$};
}    

    \draw[thick, ->] (P-00) to (P-01);

% edges ``solid''
\foreach \xx/\yy[evaluate=\xx as \nxx using int(\xx+1),evaluate=\yy as \nyy using int(\yy+1)] in {0/0,1/1,2/2,0/2,2/0,1/1}
{
    \draw[thick, ->] (P-\xx\yy) to (P-\nxx\yy);
    \draw[thick, ->] (P-\xx\yy) to (P-\xx\nyy);
    \draw[thick, ->] (P-\xx\yy) to (U-\xx\yy);
    \draw[dashed, ->] (U-\xx\yy) to (P-\nxx\nyy);    
    \draw[dashed, ->] (U-\xx\yy) to (S-\xx\yy);    
    \draw[dashed, ->] (S-\xx\yy) to (T-\xx\yy);    
    \draw[dashed, ->] (U-\xx\yy) to (T-\xx\yy);    
    \draw[dashed, ->] (P-\nxx\yy) to (P-\nxx\nyy);    
    \draw[dashed, ->] (P-\xx\nyy) to (P-\nxx\nyy);    
}

\foreach \xx/\yy[evaluate=\xx as \nxx using int(\xx+1),evaluate=\yy as \nyy using int(\yy+1)] in {0/1,1/0,1/2,2/1}
{
    \draw[dashed, ->] (P-\xx\yy) to (P-\nxx\yy);
    \draw[dashed, ->] (P-\xx\yy) to (P-\xx\nyy);
    \draw[dashed, ->] (P-\xx\yy) to (U-\xx\yy);
    \draw[thick, ->] (U-\xx\yy) to (P-\nxx\nyy);  
    \draw[thick, ->] (U-\xx\yy) to (S-\xx\yy);    
    \draw[thick, ->] (S-\xx\yy) to (T-\xx\yy);    
    \draw[thick, ->] (U-\xx\yy) to (T-\xx\yy);      
    \draw[thick, ->] (P-\nxx\yy) to (P-\nxx\nyy);    
    \draw[thick, ->] (P-\xx\nyy) to (P-\nxx\nyy);    

}

\end{scope}
\end{tikzpicture}
 \caption{A model of $\Phi$.}\label{f:model}
\end{figure}

The crucial property of the formula $\Phi$ is that the structure presented in Figure~\ref{f:model}, as well as similar, larger ``grid-like'' structures, is a model of $\Phi$. It is worth to notice that the potential  transitivity  of $R_1$ and $R_2$ plays a minor role in this structure -- the longest $R_1$ and $R_2$ paths are of length $2$:  $U_{x,y}$, $S_{x,y}$, $T_{x,y}$, and there are no other such paths of length $2$.

Let $\klasa$ be the set of frames defined by $\Phi$. To show that the problems $\klasa$-GSAT and $\klasa$-GFINSAT are undecidable, one can encode the domino problem \cite{Berger66} in a standard manner (c.f. \cite{MichaliszynOtop15}).

To obtain the result for SAT and FINSAT, we define $\Phi'$ as the conjunction of $\Phi$ and the following conjunct:
\begin{enumerate}
\item[5.] $\forall x, y, z, v . x R_1 y \wedge x R_2 y \wedge z R_1 v \Rightarrow x R_1 v$
\end{enumerate}

Observe that the intended model of $\Phi$, exemplified in Figure~\ref{f:model} is also a model of $\Phi'$ (as every world has only $R_1$-successors or only $R_2$-successors), so the proof for GSAT and GFINSAT also works for $\Phi'$.

Now, to obtain a proof for the local satisfiability problem, we use the modal formula $\phi_D$ from the proof for the global satisfiability problem to define the formula $\phi_D'$ as follows:
\[
\Diamond_1 \top \wedge \Diamond_2 \top \wedge \square_1 \phi_D 
\]
It requires that the initial world has both a $R_1$-successor and a $R_2$-successor, which means that by $\Phi'$ it has to be $R_1$ or $R_2$-connected to every world of the model that has a predecessor. To conclude the result it is enough to see that the words without predecessors, except from the initial one, can be removed from the model without influencing the satisfaction of the formula.

\end{proof}

\section{Discussion and conclusions}\label{s:future}\label{conclusion}
We showed that elementary multimodal logics may be undecidable, even if we only  consider universal Horn formulas, two modalities and possibly require some relations to be transitive. It is worth to notice that our first order formulas include some non-trivial interaction between the relations $R_1$ and $R_2$ -- an interesting question is whether the problem becomes decidable if we forbid such relations, i.e., require that each conjunct uses only one relation.
\bibliographystyle{plain}
\bibliography{all,bib}
\end{document}